\documentclass[orivec]{llncs}

\usepackage{longtable}
\usepackage{tabularx}
\usepackage{amsmath,amssymb}
\usepackage{mathtools}
\usepackage{algorithm}
\usepackage{algpseudocode}
\usepackage{pifont}
\usepackage{graphicx}
\usepackage{color}
\usepackage{xcolor}
\usepackage{caption}
\usepackage{xspace}
\usepackage{tikz}
\usepackage{times}
\usepackage[scaled=.92]{helvet}
\usetikzlibrary{arrows,calc,positioning,shapes.gates.logic.US,automata}
\usepgflibrary{shapes.multipart}
\usepackage{listings}
\usepackage{wrapfig}

\usepackage{pgfplots}
\usepackage{wrapfig}
\usepackage{booktabs}
\usepackage{hyperref}
\usepackage{pdflscape}
\usepackage{wasysym}
\usepackage{subcaption}
\usepackage{placeins}

\usepackage{thm-restate}

\newcommand{\true}{\ensuremath{\mathsf{true}}\xspace}

\newcommand{\defn}{\ensuremath{\stackrel{\textup{\tiny def}}{=}}}

\newcommand{\traces}{\ensuremath{\mathcal{T}}\xspace}
\newcommand{\APtraces}{\ensuremath{\mathsf{APTr}}\xspace}
\newcommand{\AP}{\ensuremath{\mathsf{AP}}\xspace}
\newcommand{\STS}{\ensuremath{\mathcal{S}}\xspace}
\newcommand{\Vars}{\ensuremath{\mathcal{V}}\xspace}
\newcommand{\InVars}{\ensuremath{\mathcal{I}}\xspace}
\newcommand{\OutVars}{\ensuremath{\mathcal{O}}\xspace}
\newcommand{\StateVars}{\ensuremath{\mathcal{X}}\xspace}
\newcommand{\restrict}[1]{\ensuremath{\vert}_{{#1}}}
\newcommand{\mutvar}{\ensuremath{\mathrm{mut}}}

\newcommand{\trans}[1]{\ensuremath{\xrightarrow{#1}}}

\author{Andreas Fellner $^{1,2}$ \and Mitra Tabaei Befrouei $^{2}$ \and Georg Weissenbacher$^{2}$ }

\institute{$^{1}$ AIT Austrian Institute of Technology, $^{2}$ TU Wien}

\title{Mutation Testing with Hyperproperties\thanks{
The research was supported by ECSEL JU
under the project H2020 737469 AutoDrive - Advancing failaware,
fail-safe, and fail-operational electronic components, systems, and
architectures for fully automated driving to make future mobility
safer, affordable, and end-user acceptable.
}}

\begin{document}

\maketitle 

\begin{abstract}
We present a new method for model-based mutation-driven test case generation.
Mutants are generated by making small syntactical modifications to the model or
source code of the system under test. A test case kills a mutant if
the behavior of the mutant deviates from the original system when
running the test.
In this work, we use hyperproperties---which allow to express
relations between multiple executions---to formalize different notions of 
\emph{killing} for both deterministic as well as non-deterministic
models. The resulting hyperproperties are universal in the
sense that they apply to arbitrary reactive models and mutants. 
Moreover, an off-the-shelf model checking tool for hyperproperties can
be used to generate 
test cases. We evaluate our approach on a number of models
expressed in two different modeling languages by generating 
tests using a state-of-the-art mutation testing tool.
\end{abstract}


\section{Introduction}
\label{sec:intro}

Mutations---small syntactic modifications of programs that mimic
typical programming errors---are used to assess the quality of
existing test suites. A test \emph{kills} a mutated program
(or \emph{mutant}), obtained by applying a \emph{mutation operator}
to a program, if its outcome for the mutant deviates
from the outcome for the unmodified program. The percentage
of mutants killed by a given test suite serves as a metric for
test quality. The approach is based on two assumptions: (a)
the \emph{competent programmer hypothesis} \cite{Budd1979},
which states that implementations are typically close-to-correct,
and (b) the \emph{coupling effect} \cite{Offutt1992},
which states that a test suites ability to detect simple errors
(and mutations) is indicative of its ability to detect complex errors. 

In the context of model-based testing, mutations are also used
to design tests. Model-based test case generation is the
process of deriving tests from a reference model (which is assumed
to be free of faults) in such a way that they reveal 
any non-conformance of the reference model and its mutants, i.e., kill the mutants.
The tests detect potential errors (modeled by mutation operators) of implementations, 
treated as a black box in this setting, 
that conform to a mutant instead of the reference model.
A test \emph{strongly}
kills a mutant if it triggers an observable difference in
behavior \cite{Budd1979}, and \emph{weakly} kills a mutant if the deviation
is merely in a difference in traversed program states \cite{Howden1982}.

The aim of our work is to automatically construct tests that strongly kill
mutants derived from a reference model. To this end,
we present two main contributions:
\begin{itemize}
\item[(1)] A formalization of mutation killing in terms of
  \emph{hyperproperties} \cite{DBLP:journals/jcs/ClarksonS10},
  a formalism to relate multiple execution traces of a program
  which has recently gained popularity due to its ability
  to express security properties such as non-interference and
  observational determinism. Notably, our formalization also
  takes into account potential non-determinism,
  which significantly
  complicates killing of mutants due to the unpredictability of the
  test outcome.
\item[(2)] An approach that enables the automated construction
  of tests by means of \emph{model checking} the proposed hyperproperties
  on a model that aggregates the reference model
  and a mutant of it. To circumvent limitations of currently available
  model checking tools for hyperproperties, we present a transformation
  that enables the control of non-determinism via additional program
  inputs. We evaluate our approach using a state-of-the-art
  model checker on a number of models expressed in two different
  modeling languages.
\end{itemize}

\paragraph{Running example.} We illustrate the main concepts of our work
in Figure~\ref{fig:beverage}. Figure~\ref{fig:beverage_code} shows the
SMV \cite{mcmillan92smv} model of a beverage machine, which
non-deterministically serves
\texttt{coff} (coffee) or \texttt{tea} after input \texttt{req} (request), 
assuming that there is still enough \texttt{wtr} (water) in the tank.
Water can be refilled with input \texttt{fill}.
The symbol $\varepsilon$ represents absence of input and output, respectively.

The code in Figure~\ref{fig:beverage_code} includes the 
variable {\tt mut} (initialized non-deter\-ministically in line 1),
which enables the activation of a mutation in line 7.
The mutant refills $1$ unit of water only, whereas the original
model fills $2$ units.

Figure~\ref{fig:beverage_hyperproperty} states a hyperproperty over
the inputs and outputs of the model formalizing that
the mutant can be killed \emph{definitely} (i.e.,
independently of non-deterministic choices). The execution
shown in Figure~\ref{fig:beverage_test} is a witness for this claim:
the test requests two drinks after filling the tank. 
For the mutant, the second request will necessarily fail, 
as indicated in Figure~\ref{fig:beverage_kill}, which shows
all possible output sequences for the given test.

\noindent
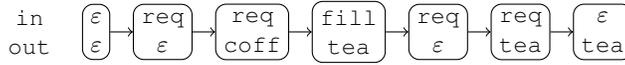
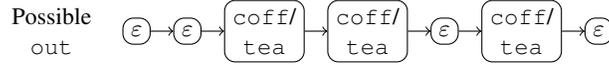
\begin{figure}[tb]
 \centering
	\begin{minipage}[b]{0.42\textwidth}
		\begin{subfigure}[b]{\textwidth}
			\begin{lstlisting}[basicstyle=\scriptsize]
init(in/out/wtr/mut):=$\varepsilon/\varepsilon$/2/$\{{\scriptstyle \top,\bot}\}$
next(in) :={$\varepsilon$,req,fill}
next(out):= 
 if(in=req&wtr>0):{coff,tea}
 else            :$\varepsilon$
next(wtr):= 
 if  (in=fill):(mut ? 1 : 2)
 elif(in=req&wtr>0):wtr-1
 else              :wtr
next(mut):=mut
			\end{lstlisting}
			\caption{Beverage machine with cond. mutant}
			\label{fig:beverage_code}
		\end{subfigure}
	\end{minipage}\hfill%
	\begin{minipage}[b]{0.58\textwidth}
		\begin{subfigure}[b]{\textwidth}
		  \scriptsize
				\begin{flalign*}
					&\exists\pi \forall \pi' \forall \pi'' \\
						 &\Box \big(\neg \texttt{mut}_{\pi} \wedge \texttt{mut}_{\pi'} \wedge \neg \texttt{mut}_{\pi''}\\
						 &\phantom{\Box \big(}[\texttt{in=}\varepsilon]_{\pi} \leftrightarrow [\texttt{in=}\varepsilon]_{\pi'} \leftrightarrow [\texttt{in=}\varepsilon]_{\pi''} \wedge \\
						 &\phantom{\Box \big(}[\texttt{in=req}]_{\pi} \leftrightarrow [\texttt{in=req}]_{\pi'} \leftrightarrow [\texttt{in=req}]_{\pi''} \wedge \\
						 &\phantom{\Box \big(}[\texttt{in=fill}]_{\pi} \leftrightarrow [\texttt{in=fill}]_{\pi'} \leftrightarrow [\texttt{in=fill}]_{\pi''}\big) \rightarrow \\
						 &\lozenge \big(\neg([\texttt{o=}\varepsilon]_{\pi'} \leftrightarrow [\texttt{o=}\varepsilon]_{\pi''}) \vee \\
						 &\phantom{\lozenge \big(} \neg([\texttt{o=coff}]_{\pi'} \leftrightarrow [\texttt{o=coff}]_{\pi''}) \vee \\
						 &\phantom{\lozenge \big(} \neg([\texttt{o=tea}]_{\pi'} \leftrightarrow [\texttt{o=tea}]_{\pi''})\big)
						\end{flalign*}
			\caption{Hyperproperty expressing killing}
			\label{fig:beverage_hyperproperty}
		\end{subfigure}
	\end{minipage}\hfill \\[.3cm]
	\begin{subfigure}[b]{\textwidth}
		\centering
		\begin{tikzpicture}[node distance = .3cm, every text node part/.style={align=center}]
			\node(expl) {\texttt{in}\\\texttt{out}};
			\node[rounded corners,rectangle,right=of expl,draw] (init) {$\varepsilon$\\$\varepsilon$};
			\node[rounded corners,rectangle,right=of init,draw] (t1) {\texttt{req}\\$\varepsilon$};
			\node[rounded corners,rectangle,right=of t1,draw] (t2) {\texttt{req}\\\texttt{coff}};
			\node[rounded corners,rectangle,right=of t2,draw] (t3) {\texttt{fill}\\\texttt{tea}};
			\node[rounded corners,rectangle,right=of t3,draw] (t4) {\texttt{req}\\$\varepsilon$};
			\node[rounded corners,rectangle,right=of t4,draw] (t5) {\texttt{req}\\\texttt{tea}};
			\node[rounded corners,rectangle,right=of t5,draw] (t6) {$\varepsilon$\\\texttt{tea}};
			\draw[->] (init) -- (t1);
			\draw[->] (t1) -- (t2);
			\draw[->] (t2) -- (t3);
			\draw[->] (t3) -- (t4);
			\draw[->] (t4) -- (t5);
			\draw[->] (t5) -- (t6);
		\end{tikzpicture}%
		\vspace{.5em}%
		\caption{Definitely killing test}
		\label{fig:beverage_test}
	\end{subfigure} \\[.3cm]
	\begin{subfigure}[b]{\textwidth}
		\centering
		\begin{tikzpicture}[node distance = .3cm, every text node part/.style={align=center}]
			\node (oexpl) {Possible\\\texttt{out}};
			\node[rounded corners,rectangle,right=of oexpl,draw] (oinit) {$\varepsilon$};
			\node[rounded corners,rectangle,right=of oinit,draw] (ot1) {$\varepsilon$};
			\node[rounded corners,rectangle,right=of ot1,draw] (ot2) {\texttt{coff}/ \\\texttt{tea}};
			\node[rounded corners,rectangle,right=of ot2,draw] (ot3) {\texttt{coff}/ \\\texttt{tea}};
			\node[rounded corners,rectangle,right=of ot3,draw] (ot4) {$\varepsilon$};
			\node[rounded corners,rectangle,right=of ot4,draw] (ot5) {\texttt{coff}/ \\\texttt{tea}};
			\node[rounded corners,rectangle,right=of ot5,draw] (ot6) {$\varepsilon$};
			\draw[->] (oinit) -- (ot1);
			\draw[->] (ot1) -- (ot2);
			\draw[->] (ot2) -- (ot3);
			\draw[->] (ot3) -- (ot4);
			\draw[->] (ot4) -- (ot5);
			\draw[->] (ot5) -- (ot6);			
		\end{tikzpicture}%
		\vspace{.5em}%
		\caption{Spurious test response of mutant}
		\label{fig:beverage_kill}
	\end{subfigure}
	\caption{Beverage machine running example}
	\label{fig:beverage}
\end{figure}

\paragraph{Outline.} Section \ref{sec:Hyper-logic} introduces
our system model and HyperLTL. Section \ref{sec:formalization-kr}
explains the notions of \emph{potential} and \emph{definite} killing of mutants,
which are then formalized in terms of hyperproperties
for deterministic and non-deterministic models in Section
\ref{sec:hyperkilling}. Section \ref{sec:tcg} introduces a
transformation to control non-determinism in models, and
Section \ref{sec:experiments} describes our experimental results.
Related work is discussed in Section \ref{sec:rel-work}.


\section{Preliminaries}
\label{sec:Hyper-logic}

This section introduces symbolic transition systems as our formalisms
for representing discrete reactive systems and provides the syntax and
semantics of HyperLTL, a logic for hyperproperties.

\subsection{System Model}
\label{sec:sysmod}


A symbolic transition system (STS) is a tuple
$\STS = \langle \InVars,\OutVars,\StateVars,\alpha,\delta \rangle$, 
where $\InVars,\OutVars,\StateVars$ are finite sets of input, output, and state variables,
$\alpha$ is a formula over $\StateVars \cup \OutVars$ (the initial conditions predicate),
and $\delta$ is a formula over $\InVars \cup \OutVars \cup \StateVars \cup \StateVars'$ (the transition relation predicate), 
where $\StateVars' = \{x' \mid x \in \StateVars\}$ is a set of primed variables
representing the successor states.
An input $I$, output $O$, state $X$, and successor state $X'$, respectively,
is a mapping of $\InVars,\OutVars$, $\StateVars$, and $\StateVars'$, respectively,
to values in a fixed domain that includes the elements $\top$ and $\bot$ (representing true and false, respectively). $Y\restrict{\Vars}$ denotes the restriction of the domain of mapping $Y$ to the variables $\Vars$.
Given a valuation $Y$ and a Boolean variable
$v\in\Vars$, $Y(v)$ denotes the value of $v$ in $Y$ (if defined)
and  $Y[v]$ and $Y[\neg v]$ denote $Y$ with $v$ set to $\top$
and $\bot$, respectively.

We assume that the initial conditions- and transition relation predicate are defined in a logic that includes standard Boolean operators $\neg$,
$\wedge$, $\vee$,
$\rightarrow$, and $\leftrightarrow$. We omit further details, as
as our results do not depend on a specific formalism.
We write $X,O \models \alpha$ and $I,O,X,X' \models \delta$ to denote that
$\alpha$ and $\delta$ evaluate to true under an evaluation
of inputs $I$, outputs $O$, states $X$, and successor states $X'$.
We assume that every STS has a distinct output $O_{\varepsilon}$, representing
absence of output.

A state $X$ with output $O$ such that $X,O \models \alpha$ are an \emph{initial state} and \emph{initial output}.
A state $X$ has a transition with input $I$ to its \emph{successor state}
$X'$ with output $O$ iff $I,O,X,X' \models \delta$, denoted by
$X \trans{I,O} X'$.
A \emph{trace} of $\STS$ is a sequence of tuples of concrete inputs,
outputs, and states
$\left\langle (I_0,O_0,X_0), (I_1,O_1,X_1), (I_2,O_2,X_2), \ldots \right\rangle$
such that $X_0,O_0\models\alpha$ and
$\forall j \geq 0\,.\, X_j \trans{I_j,O_{j+1}} X_{j+1}$.
We require that every state has at least one successor, therefore all traces of $\STS$ are infinite. 
We denote by $\traces(\STS)$ the set of all traces of $\STS$. 
Given a trace $p = \left\langle (I_0,O_0,X_0), (I_1,O_1,X_1), \ldots \right\rangle$, 
we write $p[j]$ for $(I_j,O_j,X_j)$,
$p[j,l]$ for $\langle (I_j,O_j,X_j),\ldots,(I_l,O_l,X_l) \rangle$,
$p[j,\infty]$ for $\langle (I_j,O_j,X_j),\ldots \rangle$
and $p\restrict{\Vars}$ to denote
$\left\langle ({I_0}\restrict{\Vars},{O_0}\restrict{\Vars},{X_0}\restrict{\Vars}), ({I_1}\restrict{\Vars},{O_1}\restrict{\Vars},{X_1}\restrict{\Vars}), \ldots \right\rangle$.
We lift restriction to sets of traces $T$ by defining $T\restrict{\Vars}$
as $\{p\restrict{\Vars}  \mid t \in T\}$.

\STS is \emph{deterministic} iff there is a unique pair of an initial state
and initial output and for each state $X$ and input $I$, 
there is at most one state $X'$ with output $O$,
such that $X\trans{I,O} X'$. Otherwise, the model is \emph{non-deterministic}.

In the following, we presume the existence of sets of atomic propositions
$\AP=\{\AP_{\InVars}\cup\AP_{\OutVars}\cup\AP_{\StateVars}\}$
(intentionally kept abstract)\footnote{Finite domains can be characterized
  using binary encodings; infinite domains require an extension of
  our formalism in Section \ref{subsec:hyperltl} with equality and is
  omitted for the sake of simplicity.}
serving as
labels that characterize inputs, outputs, and states (or properties thereof).

For a trace $p = \left\langle (I_0,O_0,X_0), (I_1,O_1,X_1),
\ldots \right\rangle$ the corresponding trace over $\AP$ is
$\AP(p) =
\left\langle \AP(I_0) \cup \AP(O_0) \cup \AP(X_0), \AP(I_1) \cup \AP(O_1)
\cup \AP(X_1), \ldots \right\rangle$. We lift this definition to sets of traces
by defining $\APtraces(\STS) \defn \{\AP(p) \mid p \in \traces(\STS)\}$.

\begin{example}
\label{example:sts}
Figure~\ref{fig:beverage_code} shows the 
formalization of a beverage machine in SMV~\cite{mcmillan92smv}.
In Figure~\ref{fig:beverage_hyperproperty}, we use atomic propositions
to enumerate the possible values of \texttt{in} and \texttt{out}.
This SMV model closely corresponds to an STS:
the initial condition predicate $\alpha$ and transition relation
$\delta$ are formalized using integer arithmetic as follows:

{\small
\begin{equation*}
\begin{split}
\alpha \defn &\texttt{out=}\varepsilon \wedge \texttt{wtr=2} \\
\delta \defn &\texttt{wtr>0} \wedge \texttt{in=req} \wedge \texttt{out=coff} \wedge \texttt{wtr'=wtr-1} \vee \\
&\texttt{wtr>0} \wedge \texttt{in=req} \wedge \texttt{out=tea} \wedge \texttt{wtr'=wtr-1} \vee  \\
&\texttt{in=fill} \wedge \neg \texttt{mut} \wedge \texttt{out=}\varepsilon \wedge \texttt{wtr'=2}  \vee \\
&\texttt{in=fill} \wedge \texttt{mut} \wedge \texttt{out=}\varepsilon \wedge \texttt{wtr'=1} \vee \\
&\texttt{in=}\varepsilon \wedge \texttt{out=}\varepsilon \wedge \texttt{wtr'=wtr}
\end{split}
\end{equation*}}

The trace $p = \langle (\varepsilon,\varepsilon,2),(\texttt{req},\varepsilon,\texttt{2}),(\texttt{req},\texttt{coff},1)$, $(\varepsilon,\texttt{tea},0),$ $\ldots \rangle$ is one possible execution of the system (for brevity, variable names are omitted).
Examples of atomic propositions for the system are $[\texttt{in=coff}],[\texttt{out=}\varepsilon],[\texttt{wtr>0}],[\texttt{wtr=0}]$ and the respective atomic proposition trace of $p$ is 
$\AP(p) = \langle \{[\texttt{in=}\varepsilon],[\texttt{out=}\varepsilon],[\texttt{wtr>0}]\},$\\$
\{[\texttt{in=req}],[\texttt{out=}\varepsilon],[\texttt{wtr>0}]\},\{[\texttt{in=req}],[\texttt{out=coff}],[\texttt{wtr>0}]\},
\{[\texttt{in=req}],$\\$[\texttt{out=tea}],[\texttt{wtr=0}]\} \ldots \rangle$
\end{example}

\subsection{HyperLTL}
\label{subsec:hyperltl}

In the following, we provide an overview of the HyperLTL, a logic for hyperproperties, 
sufficient for understanding the formalization in Section \ref{sec:hyperkilling}.
For details, we refer the reader to \cite{Clarkson2014}.
HyperLTL is defined over atomic proposition traces (see
Section \ref{sec:sysmod}) of a fixed STS
$\STS = \langle \InVars,\OutVars,\StateVars,\alpha,\delta \rangle$ as defined
in Section \ref{sec:sysmod}.

\paragraph{Syntax.} 
Let $\AP$ be a set of atomic propositions and
let $\pi$ be a \emph{trace variable} from 
a set $\mathcal{V}$ of trace variables.
Formulas of HyperLTL are defined by the following grammar:

\[
\begin{array}{lcccccccccccc}
\psi & ::= &\exists \pi. \psi &|& \forall\pi. \psi &|& \varphi &&&&&\\
\varphi & ::= &a_{\pi} &|& \neg \varphi &|& \varphi \vee \varphi &|&
\Circle \varphi &|& \varphi \, \mathcal{U} \varphi
\end{array}\]

Connectives $\exists$ and $\forall$ are universal and existential trace quantifiers, read as "along some traces" and "along all traces".
In our setting, atomic propositions $a \in \AP$ express facts about
states or the presence of inputs and outputs. 
Each atomic proposition is sub-scripted with a trace variable to indicate
the trace it is associated with.
The Boolean connectives $\wedge$, $\rightarrow$, and $\leftrightarrow$
are defined in terms of $\neg$ and $\vee$ as usual.
Furthermore, we use the standard temporal operators
\emph{eventually} $\diamondsuit \varphi \defn \true \,\, \mathcal{U} \varphi$, and 
\emph{always} $\Box \varphi \defn \neg \diamondsuit \neg \varphi$.

\paragraph{Semantics}

$\Pi \models_{\STS} \psi$ states that $\psi$ is valid for
a given mapping $\Pi: \mathcal{V} \rightarrow \APtraces(\STS)$
of trace variables to atomic proposition traces.
Let $\Pi\left[\pi \mapsto p \right]$ be as $\Pi$
except that $\pi$ is mapped to $p$.
We use $\Pi\left[i, \infty \right]$ to denote the
trace assignment $\Pi'(\pi)=\Pi(\pi)\left[i, \infty \right]$ for all $\pi$.
The validity of a formula is defined as follows:
\[
\begin{array}{lcccl}
\Pi \models_{\STS}  a_{\pi} &&  \mbox{iff} && a \in \Pi(\pi)[0]  \\
\Pi \models_{\STS} \exists \pi.\psi  && \mbox{iff}  &&  \mbox{there exists } p \in \APtraces(\STS): \Pi\left[\pi \mapsto p \right] \models _{\STS} \psi   \\
\Pi \models_{\STS} \forall \pi.\psi  && \mbox{iff} && \mbox{for all }  p \in \APtraces(\STS): \Pi\left[\pi \mapsto p \right] \models_{\STS} \psi   \\
\Pi \models_{\STS}  \neg \varphi &&  \mbox{iff} &&  \Pi \not\models_{\STS} \varphi   \\
\Pi \models_{\STS}  \psi_1 \vee \psi_2 &&  \mbox{iff}  && \Pi \models_{\STS} \psi_1  \mbox{ or } \Pi \models_{\STS} \psi_2   \\
\Pi \models_{\STS}  \Circle \varphi &&  \mbox{iff} && \Pi\left[1, \infty \right] \models_{\STS} \varphi   \\
\Pi \models_{\STS}  \varphi_1 \, \mathcal{U} \varphi_2 &&  \mbox{iff}  && \mbox{there exists }  i \geq 0:   \Pi\left[i, \infty \right] \models_{\STS} \varphi_2  \\
		&& &&  \mbox{and for all } 0 \leq j < i \mbox{  we have  } \Pi\left[j, \infty \right] \models_{\STS} \varphi_1  \\
\end{array}
\]

We write $\models_{\STS} \psi$ if
$\Pi \models_{\STS} \psi$ holds and $\Pi$ is empty.
We call $q \in \traces(\STS)$ a $\pi$-witness of a formula
$\exists \pi.\psi$, if $\Pi\left[\pi \mapsto p \right] \models_{\STS} \psi$
and $\AP(q) = p$.


\section{Killing mutants}
\label{sec:formalization-kr}

In this section, we introduce mutants, tests, and the notions of potential and definite killing.
We discuss how to represent an STS and its corresponding mutant 
as a single STS, which can then be model checked to determine killability.

\subsection{Mutants}

Mutants are variations of a model $\STS$ obtained by applying small
modifications to the syntactic representation of $\STS$.
A mutant of an STS $\STS = \langle \InVars,\OutVars,\StateVars,\alpha,\delta \rangle$ (the \emph{original model})
is an
STS $\STS^m = \langle \InVars,\OutVars,\StateVars,\alpha^m,\delta^m \rangle$ 
with equal sets of input, output, and state variables as $\STS$ but a
deviating initial predicate and/or transition relation.
We assume that $\STS^m$ is equally input-enabled as $\STS$, 
that is $\traces(\STS^m)\restrict{\InVars} = \traces(\STS)\restrict{\InVars}$, 
i.e., the mutant and model accept the same sequences of inputs.
In practice, this can easily be achieved by using self-loops with
empty output to ignore unspecified inputs.
We use standard mutation operators, such as disabling transitions, replacing operators, etc.
Due to space limitations and the fact that mutation operators are not the primary focus of this work, 
we do not list them here, but refer to Appendix~\ref{appendix:proofs} and \cite{Arcaini2015}.

We combine an original model represented by $\STS$ and a mutant $\STS^m$ into a \emph{conditional mutant} $\STS^{c(m)}$, 
in order to perform mutation analysis via model checking the combined model.
The conditional mutant is defined as 
$\STS^{c(m)}~\defn~\langle\InVars,\OutVars,\StateVars~\cup~\{\mutvar\},\alpha^{c(m)},\delta^{c(m)}\rangle$,
where $\mutvar$ is a fresh Boolean variable used to distinguish
states of the original and the mutated STS.

Suppose $\STS^m$ replaces a sub-formula $\delta_0$ of $\delta$ by $\delta_0^m$,
then the transition relation predicate of the conditional mutant $\delta^{c(m)}$ is obtained by replacing 
$\delta_0$ in $\delta$ by $(\mutvar \wedge \delta_0^m) \vee
(\neg\mutvar \wedge \delta_0)$.
We fix the value of $\mutvar$ in transitions
by conjoining $\delta$ with $\mutvar \leftrightarrow \mutvar'$.
The initial conditions predicate of the conditional mutant is defined similarly.
Consequently, for a trace $p \in \traces(\STS^{c(m)})$ it holds that 
if $p|_{\{\mutvar\}}=\{\bot\}^{\omega}$ then
$p\restrict{\InVars\cup\OutVars\cup\StateVars} \in \traces(\STS)$,
and if $p|_{\{\mutvar\}}=\{\top\}^{\omega}$ then
$p\restrict{\InVars\cup\OutVars\cup\StateVars} \in \traces(\STS^m)$.
Formally, $\STS^{c(m)}$ is non-deterministic, since $\mutvar$ is
chosen non-deterministically in the initial state. However,
we only refer to $\STS^{c(m)}$ as ``non-deterministic'' if
either $\STS$ or $\STS^m$ is non-deterministic, as $\mutvar$ is
typically fixed in the hypertproperties in Section \ref{sec:hyperkilling}.

%
%
Example \ref{example:sts} and Figure~\ref{fig:beverage_code} show
a conditional mutant as an STS and in SMV.

\subsection{Killing}
\label{sec:killing}

Killing a mutant amounts to finding inputs for which the mutant
produces outputs that deviate from the original model.
In a reactive, model-based setting, killing has been formalized
using conformance relations \cite{Tretmans1996},
for example in \cite{aichernig_refinement_2014,Fellner2017}, where
an implementation \emph{conforms} to its specification if all its
input/output sequences are part of/allowed by the specification.
In model-based testing, the model takes the role of the specification and
is assumed to be correct by design.
The implementation is treated as black box, and therefore
mutants of the specification serve as its proxy.
Tests (i.e., input/output sequences) that demonstrate
non-conformance between the model and its mutant
can be used to check whether the implementation adheres
to the specification or contains the bug reflected in the mutant.
The execution of a test on a system under test fails if the sequence of
inputs of the test triggers a sequence of outputs that 
deviates from those predicted by the test.
Formally, tests are defined as follows:

\begin{definition}[Test]
	A \emph{test} $t$ of \emph{length $n$} for $\STS$ comprises
        inputs $t\restrict{\InVars}$ and outputs $t\restrict{\OutVars}$
        of length $n$, such that there exists a trace
        $p \in \traces(\STS)$ with
        $p\restrict{\InVars}[0,n]~=~t\restrict{\InVars}$
        and $p\restrict{\OutVars}[0,n]~=~t\restrict{\OutVars}$.
\end{definition}

For non-deterministic models, in which a single sequence of
inputs can trigger different sequences of outputs,
we consider two different notions of killing. We say that a
mutant can be \emph{potentially killed} if there exist inputs
for which the mutant's outputs deviate from the original model
given an appropriate choice of non-deterministic initial states
and transitions. In practice, executing a test that potentially
kills a mutant on a faulty implementation that exhibits non-determinism
(e.g., a multi-threaded program)
may fail to demonstrate non-conformance (unless the non-determinism
can be controlled).
A mutant can be \emph{definitely killed} if
there exists a sequence of inputs for which the behaviors
of the mutant and the original model deviate independently
of how non-determinism is resolved.
Note potential and definite killability are orthogonal to the
folklore notions of weak and strong killing, which capture
different degrees of observability.
Formally, we define potential and definite killability as follows:

\begin{definition}[Potentially killable]
\label{def:potentiallykillable}
$\STS^m$ is \emph{potentially killable} if 
$$\traces(\STS^m)\restrict{\InVars\cup \OutVars} \nsubseteq
\traces(\STS)\restrict{\InVars\cup \OutVars}$$

\noindent Test $t$ for $\STS$ of length $n$ \emph{potentially kills} $\STS^m$ if 
$$\{q[0,n] \mid q \in \traces(\STS^m) \wedge q[0,n]\restrict{\InVars} =
t\restrict{\InVars}\}\restrict{\InVars\cup \OutVars}
\nsubseteq \{p[0,n] \mid p \in \traces(\STS)\}\restrict{\InVars\cup \OutVars}.$$

\end{definition}

\begin{definition}[Definitely killable]
\label{def:definitelykillable}
$\STS^m$ is \emph{definitely killable} if there is a sequence of inputs $\vec{I} \in \traces(\STS)\restrict{\InVars}$, such that 
$$\{q \in \traces(\STS^m) \mid q\restrict{\InVars} = \vec{I}\}\restrict{\OutVars} \cap \{p \in \traces(\STS) \mid p\restrict{\InVars} =  \vec{I}\}\restrict{\OutVars} = \emptyset$$

\noindent Test $t$ for $\STS$ of length $n$ \emph{definitely kills} $\STS^m$ if 
\begin{equation*}
\begin{split}
&\{q[0,n] \mid q \in \traces(\STS^m) \wedge q[0,n]\restrict{\InVars} = t\restrict{\InVars}\}\restrict{\OutVars} \cap \\
&\{p[0,n] \mid p \in \traces(\STS) \wedge p[0,n]\restrict{\InVars} = t\restrict{\InVars}\}\restrict{\OutVars} = \emptyset
\end{split}
\end{equation*}
\end{definition}

\begin{definition}[Equivalent Mutant]
$\STS^m$ is \emph{equivalent} iff $\STS^m$ is not potentially killable.
\end{definition}

Note that definite killability is stronger than potential killabilty,
though for deterministic systems, the two notions coincide.

\begin{restatable}{proposition}{killability}
\label{prop:killability}
If $\STS^m$ is definitely killable then $\STS^m$ is potentially killable.\\
If $\STS^m$ is deterministic then: $\STS^m$ is potentially killable iff $\STS^m$ is definitely killable.
\end{restatable}

The following example shows a definitely killable mutant, a mutant that is only potentially killable, and an equivalent mutant.

\begin{example}
\label{ex:potentially_definitely}

The mutant in Figure~\ref{fig:beverage_code}, 
is definitely killable, since we can force the system into a state 
in which both possible outputs of the original system
(\texttt{coff}, \texttt{tea}) differ from the only possible output
of the mutant ($\varepsilon$).

Consider a mutant that introduces non-determinism by replacing
line 7 with the code \texttt{\textbf{if}(in=fill):(mut ? \{1,2\} : 2)},
indicating that the machine is filled with either $1$ or $2$ units of water.
This mutant is potentially but not definitely killable, as
only one of the non-deterministic choices leads to a deviation of the outputs.

Finally, consider a mutant that replaces line 4 
with \texttt{\textbf{if}(in=req\&wtr>0):(mut ? coff : \{coff,tea\}})
and removes the \texttt{mut} branch of line 7, yielding a machine
that always creates coffee.
Every implementation of this mutant is also correct with
respect to the original model. Hence,
we consider the mutant equivalent, even though the original model,
unlike the mutant, can output \texttt{tea}.

\end{example}


\section{Killing with hyperproperties}
\label{sec:hyperkilling}

In this section, we provide a formalization of potential and definite killability in terms of
HyperLTL, assert the correctness of our formalization with respect to
Section \ref{sec:formalization-kr}, and explain how tests
can be extracted by model checking the HyperLTL properties.
All HyperLTL formulas depend on inputs and outputs of the model, 
but are model-agnostic otherwise.
The idea of all presented formulas is to discriminate between traces of
the original model ($\Box \neg \mutvar_{\pi}$) and traces of the mutant
($\Box\mutvar_{\pi}$).
Furthermore, we quantify over pairs $(\pi,\pi')$ of traces with globally equal inputs $(\Box \bigwedge_{i \in \AP_{\InVars}} i_{\pi} \leftrightarrow i_{\pi'})$ 
and express that such pairs will eventually have different outputs $(\lozenge \bigvee_{o \in \AP_{\OutVars}} \neg( o_{\pi} \leftrightarrow o_{\pi'}))$.

\subsection{Deterministic Case}
To express killability (potential and definite) of a deterministic model and mutant, we need to find a trace of the model ($\exists \pi$)
such that the trace of the mutant with the same inputs ($\exists \pi'$) eventually diverges in outputs, formalized by $\phi_1$ as follows:
\begin{equation*}
	\phi_1(\InVars,\OutVars) := \exists \pi \exists \pi' \Box (\neg \mutvar_{\pi} \wedge \mutvar_{\pi'} \bigwedge_{i \in \AP_{\InVars}} i_{\pi} \leftrightarrow i_{\pi'}) \wedge \lozenge (\bigvee_{o \in \AP_{\OutVars}} \neg( o_{\pi} \leftrightarrow o_{\pi'} ) ) \label{eq:detkill}
\end{equation*}

\begin{restatable}{proposition}{dkill}
	\label{thm:det}
	For a deterministic model $\STS$ and mutant $\STS^m$ it holds that
	$$\STS^{c(m)} \models \phi_1(\InVars,\OutVars) \text{ iff } \STS^m \text{ is killable}.$$

	\noindent If $t$ is a $\pi$-witness for $\STS^{c(m)} \models \phi_1(\InVars,\OutVars)$, then
 $t[0,n]\restrict{\InVars \cup \OutVars}$ kills $\STS^m$ (for some $n\in\mathbb{N}$).
	
\end{restatable}

\subsection{Non-deterministic Case}

For potential killability of non-deterministic models and mutants,%
\footnote{Appendix \ref{appendix:proofs} covers
  deterministic models with non-deterministic mutants and vice-versa.}
we need to find a trace of the mutant ($\exists \pi$) such that all traces of the model with the same inputs ($\forall \pi'$) eventually diverge in outputs, expressed in $\phi_2$:

\begin{equation*}
	 \phi_2(\InVars,\OutVars) := \exists \pi \forall \pi' \Box (\mutvar_{\pi} \wedge \neg \mutvar_{\pi'}	\bigwedge_{i \in \AP_{\InVars}} i_{\pi} \leftrightarrow i_{\pi'}) \rightarrow \lozenge (\bigvee_{o \in \AP_{\OutVars}} \neg (o_{\pi} \leftrightarrow o_{\pi'})) \label{eq:potentialkill}
\end{equation*}

\begin{restatable}{proposition}{ndpotential}
\label{thm:ndet_ndet_potential}
	For non-deterministic $\STS$ and $\STS^m$, it holds that
	$$\STS^{c(m)} \models \phi_2(\InVars,\OutVars) \text{ iff } \STS^m \text{ is potentially killable.}$$
	
		\noindent If $s$ is a $\pi$-witness for $\STS^{c(m)} \models \phi_2(\InVars,\OutVars)$, 
		then for any trace $t \in \traces(\STS)$ with $t\restrict{\InVars} = s\restrict{\InVars}$, 
		$t[0,n]\restrict{\InVars \cup \OutVars}$ potentially kills $\STS^m$ (for some $n\in\mathbb{N}$).

\end{restatable}

To express definite killability, we need to find a sequence of inputs of the model ($\exists \pi$) and 
compare all non-deterministic outcomes of the model ($\forall \pi'$) 
to all non-deterministic outcomes of the mutant ($\forall \pi''$) for these inputs, as formalized by $\phi_3$:

\begin{equation*}
\begin{split}
	\phi_3(\InVars,\OutVars) := \exists \pi \forall \pi' \forall \pi'' \Box \big(&\neg \mutvar_{\pi} \wedge \mutvar_{\pi'} \wedge \neg \mutvar_{\pi''}  \\
	& \bigwedge_{i \in \AP_{\InVars}} i_{\pi} \leftrightarrow i_{\pi'} \wedge i_{\pi} \leftrightarrow i_{\pi''}\big) \rightarrow 
	\lozenge \big(\bigvee_{o \in \AP_{\OutVars}} \neg(o_{\pi'} \leftrightarrow o_{\pi''})\big)
\end{split} \label{eq:definitekill}
\end{equation*}

In Figure~\ref{fig:beverage_hyperproperty}, we present an instance of $\phi_3$ for our running example.

\begin{restatable}{proposition}{nddefinite}
\label{thm:ndet_ndet_definite}
	For non-deterministic $\STS$ and $\STS^m$, it holds that
	$$\STS^{c(m)} \models \phi_3(\InVars,\OutVars) \text{ iff } \STS^m \text{ is definitely killable.}$$
	
	\noindent If $t$ is a $\pi$-witness for $\STS^{c(m)} \models \phi_3(\InVars,\OutVars)$, then
 $t[0,n]\restrict{\InVars \cup \OutVars}$ definitely kills $\STS^m$ (for some $n\in\mathbb{N}$).
\end{restatable}

To generate tests, we use model checking to verify whether the
conditional mutant satisfies the appropriate HyperLTL formula presented above
and obtain test cases as finite prefixes of witnesses for satisfaction.

\section{Non-deterministic models in practice}
\label{sec:tcg}


As stated above, checking the validity of the
hyperproperties in Section \ref{sec:hyperkilling} for
a given model and mutant enables test-case generation. To the
best of our knowledge, {\sc MCHyper}~\cite{Finkbeiner2015} is the only currently
available HyperLTL model checker.
Unfortunately, {\sc MC\-Hyper} is unable to model check
formulas with alternating quantifiers.%
\footnote{While \emph{satisfiability} in the presence of quantifier alternation is supported to some extent \cite{FinkbeinerHH18}.}
Therefore, we are currently limited to checking $\phi_1(\InVars,\OutVars)$
for deterministic models, since witnesses of $\phi_1$
may not satisfy $\phi_2$ in the presence of non-determinism.

To remedy this issue, we propose a transformation that
makes non-determinism \emph{controllable} by means of
additional inputs and yields a deterministic STS. 
The transformed model over-approximates killability in the
sense that the resulting test cases only kill the original
mutant if non-determinism can also be controlled in the system
under test. However, if equivalence can be established for the transformed model,
then the original non-deterministic mutant is also equivalent.

\subsection{Controlling non-determinism in STS}
\label{subsec:determinization}

The essential idea of our transformation is to introduce a fresh
input variable that enables the control of non-deterministic
choices in the conditional mutant $\STS^{c(m)}$. The new input is
used carefully to ensure that choices are consistent for the model
and the mutant encoded in $\STS^{c(m)}$. W.l.o.g., we introduce
an input variable $nd$ with a domain sufficiently large to
encode the non-deterministic choices in $\alpha^{c(m)}$ and
$\delta^{c(m)}$, and write $nd(X,O)$ to denote a value of
$nd$ that uniquely corresponds to state $X$ with output $O$.
Moreover, we add a fresh Boolean variable $x^{\tau}$ to $\StateVars$ 
used to encode a fresh initial state.

Let $\StateVars_+ \defn \StateVars \cup \{\mutvar\}$ and
$X_{+},X'_{+},I,O$ be valuations of
$\StateVars_{+}$, $\StateVars'_{+}$, $\InVars$, and $\OutVars$,
and $X$ and $X'$ denote ${X_+}\restrict{\StateVars}$ and
${X_+'}\restrict{\StateVars'}$, respectively.
Furthermore, $\psi(X)$, $\psi(X_+,I)$, and $\psi(O,X'_+)$
are formulas uniquely satisfied by 
$X$, $(X_+,I)$, and $(O,X'_+)$ respectively.

Given conditional mutant
$\STS^{c(m)} \defn \langle \InVars,\OutVars,\StateVars_+,\alpha^{c(m)},\delta^{c(m)} \rangle$, we define its controllable counterpart
$D(\STS^{c(m)}) \defn \langle \InVars \cup \{nd\},\OutVars,\StateVars_+ \cup \{x^{\tau}\},D(\alpha^{c(m)}),D(\delta^{c(m)}) \rangle$.
We initialize $D(\delta^{c(m)}) \defn \delta^{c(m)}$ and incrementally
add constraints as described below.

\paragraph{Non-deterministic initial conditions:}
Let $X$ be an arbitrary, fixed state.
The unique fresh initial state is $X^{\tau} \defn X[x^{\tau}]$, which, together with an empty output, we enforce by the new initial conditions predicate:
$$D(\alpha^{c(m)}) \defn \psi(X^{\tau}, O_{\varepsilon})$$

We add the conjunct $\neg \psi(X^{\tau}) \rightarrow \neg {x^{\tau}}'$ to $D(\delta^{c(m)})$, 
in order to force $x^{\tau}$ evaluating to $\bot$ in all states other than $X^{\tau}$.
In addition, we add transitions from $X^{\tau}$ to all pairs of initial
states/outputs in $\alpha^{c(m)}$.
To this end, we first partition the pairs in $\alpha^{c(m)}$ into
pairs shared by and exclusive to the model and the mutant:
\begin{equation*}
\begin{split}
 J^{\cap} &\defn \{(O,X_+) \mid X,O \models \alpha^{c(m)}\}\\
J^{orig}  &\defn \{(O,X_+) \mid \neg X_+(\mutvar) \wedge (X_+,O \models \alpha^{c(m)}) \wedge (X_+[\mutvar],O \not\models \alpha^{c(m)})\}\\
J^{mut} &\defn \{(O,X_+) \mid X_+(\mutvar) \wedge (X_+,O \models \alpha^{c(m)}) \wedge (X_+[\neg{\mutvar}],O \not\models \alpha^{c(m)})\}\\
\end{split}
\end{equation*}

For each $(O,X_+) \in J^{\cap} \cup J^{mut} \cup J^{orig}$,
we add the following conjunct to $D(\delta^{c(m)})$:
$$\psi(X^{\tau}) \wedge nd(O,X) \rightarrow \psi(O,X_+')$$

In addition, for inputs $nd(O,X)$ without corresponding target state
in the model or mutant, we add conjuncts to $D(\delta^{c(m)})$ that represent
self loops with empty outputs:
\begin{equation*}
\begin{split}
&\forall (O,X_+) \in J^{orig}: \psi(X^{\tau}[\mutvar]) \wedge nd(O,X) \rightarrow \psi(O_{\varepsilon},{X^{\tau}}'[\mutvar]) \\
&\forall (O,X_+) \in J^{mut}: \psi(X^{\tau}[\neg\mutvar]) \wedge nd(O,X) \rightarrow \psi(O_{\varepsilon},{X^{\tau}}'[\neg\mutvar])
\end{split}
\end{equation*}
%

\paragraph{Non-deterministic transitions:}
Analogously to initial states, for each state/input pair, we partition
the successors into successors shared or exclusive to model or mutant:
\begin{equation*}
\begin{split}
T^{\cap}_{(X_+,I)}  & \defn \{(X_+,I,O,X_+') \mid X \trans{I,O} X'\} \\
T^{orig}_{(X_+,I)} & \defn \{(X_+,I,O,X_+') \mid \neg X_+(\mutvar) \wedge (X_+ \trans{I,O} X_+') \wedge \neg(X_+[\mutvar] \trans{I,O} X_+')\} \\
T^{mut}_{(X_+,I)} & \defn \{(X_+,I,O,X_+') \mid X_+(\mutvar) \wedge (X_+ \trans{I,O} X_+') \wedge \neg(X_+[\neg\mutvar] \trans{I,O} X_+')\} \\
\end{split}
\end{equation*}
A pair $(X_+,I)$ causes non-determinism 
if
\begin{displaymath}
|(T^{\cap}_{(X_+,I)}~\cup~T^{orig}_{(X_+,I)})\restrict{\StateVars \cup \InVars \cup \OutVars \cup \StateVars'}| > 1 \text{~or~}
|(T^{\cap}_{(X_+,I)}~\cup~T^{mut}_{(X_+,I)})\restrict{\StateVars \cup \InVars \cup \OutVars \cup \StateVars'}|>1.
\end{displaymath}

For each pair $(X_+,I)$ that causes non-determinism and each $(X_+,I,O_j,{X_{+j}'}) \in T^{\cap}_{(X_+,I)} \cup T^{mut}_{(X_+,I)} \cup T^{orig}_{(X_+,I)}$, we add the following conjunct to $D(\delta^{c(m)})$:
$$\psi(X_+,I) \wedge nd(O_j,X_j) \rightarrow \psi(O_j,{X_{+j}'})$$

Finally, we add conjuncts representing self loops with empty output
for inputs that have no corresponding transition in the model
or mutant:

\begin{equation*}
\begin{split}
&\forall (X_+,I,O_j,{X_{+j}'}) \in T^{orig}_{(X_+,I)}: \psi(X_+[\mutvar],I) \wedge nd(O_j,{X_j}) \rightarrow \psi(O_{\varepsilon},X_{+j}'[\mutvar]) \\
&\forall (X_+,I,O_j,{X_{+j}'}) \in T^{mut}_{(X_+,I)}: \psi(X_+[\neg\mutvar],I) \wedge nd(O_j,{X_j}) \rightarrow \psi(O_{\varepsilon},X_{+j}'[\neg\mutvar])
\end{split}
\end{equation*}

The proposed transformation has the following properties:

\begin{restatable}{proposition}{transformationproperties}

Let $\STS$ be a model with inputs $\InVars$, outputs $\OutVars$, and mutant $\STS^m$ then

\begin{enumerate}
	\item \label{enum:deterministic}$D(\STS^{c(m)})$ is deterministic (up to $\mutvar$).
	\item \label{enum:traceinclusion}$\traces(\STS^{c(m)})\restrict{\StateVars_+ \cup \InVars \cup \OutVars} \subseteq \traces(D(\STS^{c(m)}))[1,\infty]\restrict{\StateVars_+ \cup \InVars \cup \OutVars}$.
	\item \label{enum:equiv}$D(\STS^{c(m)}) \not\models \phi_1(\InVars,\OutVars)$ then $\STS^m$ is equivalent.
\end{enumerate}

\end{restatable}

The transformed model is deterministic, since we enforce unique initial valuations and 
make non-deterministic transitions controllable through input $nd$.
Since we only add transitions or augment existing transitions with input $nd$,
every transition $X \trans{I,O} X'$ of $\STS^{c(m)}$ is still present in
$D(\STS^{c(m)})$ (when input $nd$ is disregarded).
The potential additional traces of Item~\ref{enum:traceinclusion} originate from the $O_{\varepsilon}$-labeled transitions for non-deterministic choices present
exclusively in the model or mutant. 
These transitions enable the detection of discrepancies between model and mutant
caused by the introduction or elimination of non-determinism by the mutation.

For Item~\ref{enum:equiv} (which is a direct consequence of
Item~\ref{enum:traceinclusion}), assume that the original non-deterministic
mutant is not equivalent (i.e., potentially killable). Then
$D(\STS^{c(m)}) \models \phi_1(\InVars,\OutVars)$, and the corresponding
witness yields a test which kills the mutant assuming non-determinism
can be controlled in the system under test. Killability purported
by $\phi_1$, however, could be an artifact of the transformation:
determinization potentially deprives the model of its ability
to match the output of the mutant by deliberately choosing a
certain non-deterministic transition. In
Example~\ref{ex:potentially_definitely}, we present an equivalent mutant
which is killable after the transformation, since we will detect the
deviating output \texttt{tea} of the model and $\varepsilon$ of the mutant.
Therefore, our transformation merely allows us to provide a lower
bound for the number of equivalent non-deterministic mutants.

\subsection{Controlling non-determinism in modeling languages}

The exhaustive enumeration of states ($J$) and transitions ($T$) outlined
in Section \ref{subsec:determinization} is purely theoretical and
infeasible in practice. However, an analogous result can often be
achieved by modifying the syntactic constructs of the underlying
modeling language that introduce non-determinism. 

\begin{itemize}
\item{\em Non-deterministic assignments.} Non-deterministic
choice over a finite set of elements $\{x'_1,\ldots x'_n\}$, as
provided by SMV~\cite{mcmillan92smv}, can readily be converted
into a case-switch construct over $nd$. More generally, explicit
non-deterministic assignments ${\tt x := \star}$ to state
variables {\tt x} \cite{Nelson1989} can be controlled by assigning
the value of $nd$ to {\tt x}. 

\item{\em Non-deterministic schedulers.} Non-determinism introduced
by concurrency can be controlled by introducing input variables
that control the scheduler (as proposed in \cite{Lal2009} for bounded
context switches).
\end{itemize}

In case non-determinism arises through variables under-specified in
transition relations, these variable values can be made inputs 
as suggested by Section \ref{subsec:determinization}. In general,
however, identifying under-specified variables automatically is
non-trivial.
 
\begin{example}
%
Consider again the SMV code in Figure~\ref{fig:beverage_code}, for which non-determinism can be made controllable
by replacing line \texttt{\textbf{if}(in=req\&wtr>0):\{coff,tea\}} with
lines \texttt{\textbf{if}(nd=0\&in=req\&wtr>0):coff}, \texttt{\textbf{elif}(nd=1\&in=req\&wtr>0):tea}
and
adding \texttt{\textbf{init}(nd):=\{0,1\}}.

Similarly, the STS representation of the beverage machine, given in Example~\ref{example:sts}, can be transformed by replacing the first two rules by the following two rules:
\small{
\begin{align*}
&\texttt{nd=0} \wedge \texttt{wtr>0} \wedge \texttt{in=req} \wedge \texttt{out=coff} \wedge \texttt{wtr'=wtr-1} \vee \\
&\texttt{nd=1} \wedge \texttt{wtr>0} \wedge \texttt{in=req} \wedge \texttt{out=tea} \wedge \texttt{wtr'=wtr-1} \vee 
\end{align*}
}
\end{example}

\section{Experiments}
\label{sec:experiments}

In this section, we present an experimental evaluation of the presented methods.
We start by presenting the deployed tool-chain.
Thereafter, we present a validation of our method on one case study with another model-based mutation testing tool.
Finally, we present quantitative results on a broad range of generic models.

\subsection{Toolchain}
\label{sec:results}
\label{sec:toolchain}
Figure~\ref{fig:pipeline} shows the toolchain that we use to produce test suites for models encoded in the modeling languages Verilog and SMV.
Verilog models are deterministic while SMV models can be non-deterministic.

\begin{figure}[t]
	\centering
	\includegraphics[width=.8\textwidth]{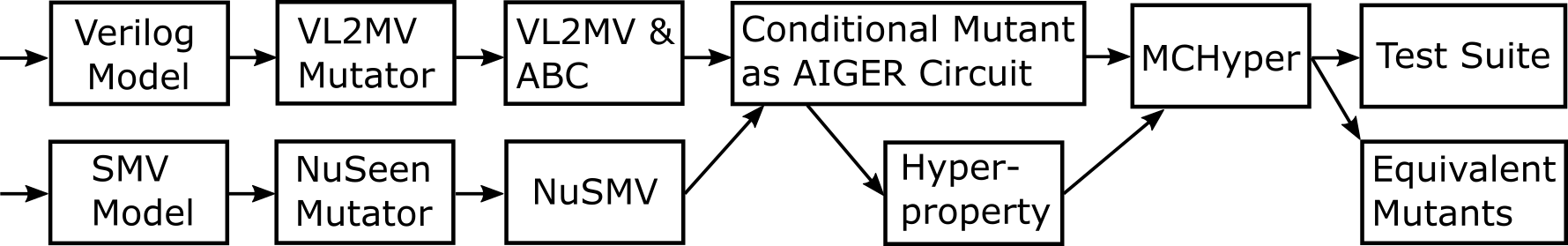}
	\caption{Tool Pipeline of our Experiments}
	\label{fig:pipeline}
\end{figure}
\noindent{\textbf{Variable annotation.}}
As a first step, we annotate variables as inputs and outputs.
These annotations were added manually for Verilog, and heuristically for SMV 
(partitioning variables into outputs and inputs).

\noindent{\textbf{Mutation and transformation.}}
We produce conditional mutants via a mutation engine.
For Verilog, we implemented our own mutation engine into the open source Verilog compiler VL2MV \cite{Cheng1993}.
We use standard mutation operators, replacing arithmetic operators, Boolean relations, Boolean connectives, constants, and assignment operators.
The list of mutation operators used for Verilog can be found in Appendix~\ref{appendix:proofs}.
For SMV models, we use the NuSeen SMV framework \cite{Arcaini2015,Arcaini2017},
which includes a mutation engine for SMV models.
The mutation operators used by NuSeen are documented in \cite{Arcaini2015}.
We implemented the transformation presented in Section~\ref{sec:tcg} into NuSeen and applied it to conditional mutants.

\noindent{\textbf{Translation.}}
The resulting conditional mutants from both modeling formalisms are translated into AIGER circuits \cite{Biere2011}.
AIGER circuits are essentially a compact representation for finite models.
The formalism is widely used by model checkers.
For the translation of Verilog models, VL2MV and the ABC model checker are used.
For the translation of SMV models, NuSMV is used.

\noindent{\textbf{Test suite creation.}}
We obtain a test suite, by model checking $\neg \phi_1(\InVars,\OutVars)$ on conditional mutants.
Tests are obtained as counterexamples, which are finite prefixes of $\pi$-witnesses to $\phi_1(\InVars,\OutVars)$.
In case we can not find a counter-example, and use a complete model checking method, the mutant is provably equivalent.

\paragraph{Case study test suite evaluation}

We compare the test suite created with our method for a case study,
with the model-based mutation testing tool MoMuT \cite{krenn_MoMuTUML_2015,Fellner2017}.
The case study is a timed version of a model of a car alarm system (CAS),
which was used in the model-based test case generation literature before 
\cite{aichernig_refinement_2014,DBLP:journals/stvr/AichernigBJKST15,Fellner2017}.

To this end, we created a test suite for a SMV formulation of the model.
We evaluated its strength and correctness on an Action System (the native modeling formalism of MoMuT) formulation of the model.
MoMuT evaluated our test suite by computing its mutation score
--- the ratio of killed- to the total number of- mutants--- with respect to 
Action System mutations, which are described in \cite{Fellner2017}.

This procedure evaluates our test suite in two ways.
Firstly, it shows that the tests are well formed, since MoMuT does not reject them.
Secondly, it shows that the test suite is able to kill mutants of a different modeling formalism than the one it was created from, 
which suggests that the test suite is also able to detect faults in implementations.

We created a test suite consisting of 61 tests, 
mapped it to the test format accepted by MoMuT.
MoMuT then measured the mutation score of our translated test suite on the Action System model, 
using Action System mutants.
The measured mutation score is 91\% on 439 Action System mutants.
In comparison, the test suite achieves a mutation score of 61\% on 3057 SMV mutants.
Further characteristics of the resulting test suite are presented in the following paragraphs.

\paragraph{Quantitative Experiments}

All experiments presented in this section were run in parallel on a machine with an Intel(R) Xeon(R) CPU at 2.00GHz, 
60 cores, and 252GB RAM.
We used 16 Verilog models which are presented in
\cite{Finkbeiner2015}, as well as models from opencores.org.
Furthermore, we used 76 SMV models that were also used in~\cite{Arcaini2015}.
Finally, we used the SMV formalism of CAS.
All models are available in \cite{benchmarks_git}.
Verilog and SMV experiments were run using property driven reachability based model checking with a time limit of 1 hour.
Property driven reachability based model checking did not perform well for CAS, 
for which we therefore switched to bounded model checking with a depth limit of 100.

\noindent{\textbf{Characteristics of models.}}
Table~\ref{tab:modelsize} present characteristics of the models.
For Verilog and SMV, we present average ($\mu$), standard deviation $(\sigma)$, minimum (Min), and maximum (Max) measures per model of the set of models.
For some measurements, we additionally present average (Avg.) or maximum (Max) number over the set of mutants per model.
We report the size of the circuits in terms of the number of inputs (\#Input), 
outputs (\#Output), state (\#State) variables as well as \emph{And} gates (\#Gates), 
which corresponds to the size of the transition relation of the model. 
Moreover, the row ``Avg. $\Delta$ \# Gates'' shows the average size difference (in $\%$ of $\#$ Gates) 
of the conditional mutant and the original model, where the average is over all mutants.
The last row of the table shows the number of the mutants that are generated for the models. 

We can observe that our method is able to handle models of respectable size, reaching thousands of gates.
Furthermore, $\Delta \#$ Gates of the conditional mutants is relatively low.
Conditional mutants allow us to compactly encode the original and mutated model in one model. 
Hyperproperties enable us to refer to and juxtapose traces from the original and mutated model, respectively. 
Classical temporal logic does not enable the comparison of different traces.
Therefore, mutation analysis by model checking classical temporal logic necessitates strictly separating traces of the original and the mutated model,
resulting in a quadratic blowup in the size of the input to the classical model-checker, compared to the size of the input to the hyperproperty model-checker.

\begin{table}%
\centering
\caption{Characteristics of Models}
\begin{tabular}{l|rrrr|rrrr|r}
\toprule
\textbf{Parameters} & \multicolumn{4}{c|}{Verilog} & \multicolumn{4}{c|}{SMV} & CAS  \\
& $\mu$ & $\sigma$ & Min & Max & $\mu$ & $\sigma$ & Min & Max &   \\
\midrule
\# Models & \multicolumn{4}{c|}{16} & \multicolumn{4}{c|}{76} &  1 \\ 
\midrule
\# Input & 186.19 & 309.59 & 4 & 949 & 8.99 & 13.42 & 0 & 88 & 58 \\
\# Output & 176.75 & 298.94 & 7 & 912 & 4.49 & 4.26 & 1 & 28 & 7 \\
\# State &  15.62 & 15.56 & 2 & 40 & - & - & - & - & - \\
\midrule
\# Gates & 4206.81 & 8309.32 & 98 & 25193 & 189.12 & 209.59 & 7 & 1015 & 1409 \\
Avg. $\Delta$ \# Gates & 3.98\% & 14.71\% & -10.2\% & 57.55\% & 8.14\% & 8.23\% & 0.22\% & 35.36\% & 0.86\% \\
\midrule
\# Mutants & 260.38 & 235.65 & 43 & 774 & 535.32 & 1042.11 & 1 & 6304 & 3057  \\
\bottomrule
\end{tabular}
\label{tab:modelsize}
\end{table}

\noindent{\textbf{Model checking results.}}
Table~\ref{tab:results} summarizes the quantitative results of our experiments.
The quantitative metrics we use for evaluating our test generation approach are
the mutation score (i.e. percentage of killed mutants) and the percentage of equivalent mutants, the number of generated tests, 
the amount of time required for generating them and the average length of the test cases. 
Furthermore, we show the number of times the resource limit was reached.
For Verilog and SMV this was exclusively the 1 hour timeout.
For CAS this was exclusively the depth limit 100.

Finally, we show the total test suite creation time, including times when reaching the resource limit.
The reported time assumes sequential test suite creation time.
However, since mutants are model checked independently, the process can easily be parallelized, which drastically reduces the total time needed to create a test suite for a model.
The times of the Verilog benchmark suite are dominated by two instances of the secure hashing algorithm (SHA), which are inherently hard cases for model checking.

We can see that the test suite creation times are in the realm of a few hours, 
which collapses to minutes when model checking instances in parallel.
However, the timing measures really say more about the underlying model checking methods than our proposed technique of mutation testing via hyperporperties.
Furthermore, we want to stress that our method is agnostic to which variant of model checking 
(e.g. property driven reachability, or bounded model checking) is used.
As discussed above, for CAS switching from one method to the other made a big difference.

The mutation scores average is around 60\% for all models.
It is interesting to notice that the scores of the Verilog and SMV models are similar on average, 
although we use a different mutation scheme for the types of models.
Again, the mutation score says more about the mutation scheme than our proposed technique.
Notice that we can only claim to report the mutation score, because, besides CAS, 
we used a complete model checking method (property driven reachability).
That is, in case, for example, 60\% of the mutants were killed and no timeouts occurred, then 40\% of the mutants are provably equivalent.
In contrast, incomplete methods for mutation analysis can only ever report lower bounds of the mutation score.
Furthermore, as discussed above, the 61.7\% of CAS translate to 91\% mutation score on a different set of mutants.
This indicates that failure detection capability of the produced test suites is well,
which ultimately can only be measured by deploying the test cases on real systems.

\vspace{-.5cm}
\begin{table}%
\centering
\caption{Experimental Results}
\begin{tabular}{l|rrrr|rrrr|r}
\toprule
\textbf{Metrics} & \multicolumn{4}{c|}{Verilog} & \multicolumn{4}{c|}{SMV} & CAS  \\
& $\mu$ & $\sigma$ & Min & Max & $\mu$ & $\sigma$ & Min & Max &   \\
\midrule
\textbf{Mutation Score} & 56.82\% & 33.1\% & 4.7\% & 99\% & 64.79\% & 30.65\% & 0\% & 100\% & 61.7 \% \\
\quad Avg. Test-case Len. & 4.26 & 1.65 & 2.21 & 8.05 & 15.41 & 58.23 & 4 & 461.52 & 5.92 \\
\quad Max Test-case Len. & 21.62 & 49.93 & 3 & 207 & 187.38 & 1278.56 & 4 & 10006 & 9 \\
Avg. Runtime & 83.08s & 267.53s & 0.01s & 1067.8s & 1.2s & 5.48s & - & 46.8s & 7.8s \\
\midrule
\textbf{Equivalent Mutants} & 33.21\% & 32.47\% & 0\% & 95.3\% & 35.21\% & 30.65\% & 0\% & 100\% & 0\% \\
Avg. Runtime& 44.77s & 119.58s & 0s & 352.2s & 0.7s & 2.02s & - & 14.9s & -  \\
\midrule
\textbf{\# Resource Limit} & 9.96\% & 27.06\% & 0\% & 86.17\% & 3.8\% & 19.24\% & 0\% & 100\% & 38.34 \% \\
\midrule
\textbf{Total Runtime} & 68.58h & 168.62h & 0h & 620.18h & 0.4h & 1.19h & 0h & 6.79h & 1.15h \\
\bottomrule
\end{tabular}
\label{tab:results}
\end{table}

\section{Related Work}
\label{sec:rel-work}

A number of test case generation techniques are based on model
checking; a survey is provided in \cite{Fraser2009}. Many of
these techniques (such as \cite{Visser2004,Rayadurgam2001,Hong2002})
differ in abstraction levels and/or coverage goals from our approach.

Model checking based mutation testing using trap properties is presented in \cite{Gargantini1999}.
Trap properties are conditions that, if satisfied, indicate a killed mutant.
In contrast, our approach directly targets the input / output behavior of the model
and does not require to formulate model specific trap properties. 

Mutation based test case generation via module checking is proposed in \cite{Boroday2007}.
The theoretical framework of this work is similar to ours,
but builds on module checking instead of hyperproperties.
Moreover, no experimental evaluation is given in this work. 

The authors of \cite{aichernig_refinement_2014} present mutation killing using SMT solving.
In this work, the model, as well as killing conditions, are encoded into a SMT formula and solved using specialized algorithms.
Similarly, the MuAlloy \cite{wang2018mualloy} framework enables model-based mutation testing for Alloy models using SAT solving.
In this work, the model, as well as killing conditions, are encoded into a SAT formula and solved using the Alloy framework.
In contrast to these approaches, we encode only the killing conditions into a formula.
This allows us to directly use model checking techniques, in contrast to SAT or SMT solving.
Therefore, our approach is more flexible and more likely to be applicable in other domains. 
We demonstrate this by producing test cases for 
models encoded in two different modeling languages. 

Symbolic methods for weak mutation coverage are proposed in \cite{DBLP:conf/icst/BardinKC14} and \cite{DBLP:conf/icst/BardinDDKPTM15}.
The former work describes the use of dynamic symbolic execution for weakly killing mutants.
The latter work describes a sound and incomplete method for detecting equivalent weak mutants.
The considered coverage criterion in both works is weak mutation, which, 
unlike the strong mutation coverage criterion considered in this work,
can be encoded as a classic safety property.
However, both methods could be used in conjunction with our method.
Dynamic symbolic execution could be used to first weakly kill mutants and thereafter strongly kill them via hyperproperty model checking.
Equivalent weak mutants can be detected with the methods of \cite{DBLP:conf/icst/BardinDDKPTM15}
to prune the candidate space of potentially strongly killable mutants for hyperpropery model checking.

A unified framework for defining multiple coverage criteria, 
including weak mutation and hyperproperties such as unique-cause MCDC, 
is proposed in \cite{DBLP:conf/icst/MarcozziDBKP17} .
While strong mutation is not expressible in this framework, 
applying hyperproperty model checking to the proposed framework is interesting future work.
\vspace{-.2cm}
\section{Conclusion}
\label{sec:conclusion}
Our formalization of mutation testing in terms of hyperproperties
enables the automated model-based generation of tests using an off-the-shelf
model checker. In particular, we study killing of mutants
in the presence of non-determinism, where test-case generation
is enabled by a transformation that makes non-determinism in models explicit
and controllable. We evaluated
our approach on publicly available SMV and Verilog models,
and will extend our evaluation to more modeling languages and models
in future work.



\clearpage

\bibliographystyle{plain}
\bibliography{paper}{}

\begin{thebibliography}{10}

\bibitem{benchmarks_git}
Mutation testing with hyperproperies benchmark models.
\newblock
  \url{https://git-service.ait.ac.at/sct-dse-public/mutation-testing-with-hyperproperties}.
\newblock Uploaded: 2019-04-25.

\bibitem{krenn_MoMuTUML_2015}
B.~Aichernig, H.~Brandl, E.~J{\"o}bstl, W.~Krenn, R.~Schlick, and S.~Tiran.
\newblock {{MoMuT}}::{{UML}} model-based mutation testing for {{UML}}.
\newblock In {\em Software Testing, Verification and Validation (ICST), 2015
  IEEE 8th International Conference on}, ICST, pages 1--8, April 2015.

\bibitem{DBLP:journals/stvr/AichernigBJKST15}
Bernhard~K. Aichernig, Harald Brandl, Elisabeth J{\"{o}}bstl, Willibald Krenn,
  Rupert Schlick, and Stefan Tiran.
\newblock Killing strategies for model-based mutation testing.
\newblock {\em Softw. Test., Verif. Reliab.}, 25(8):716--748, 2015.

\bibitem{aichernig_refinement_2014}
Bernhard~K. Aichernig, Elisabeth J{\"o}bstl, and Stefan Tiran.
\newblock {{Model-based mutation testing via symbolic refinement checking}}.
\newblock 2014.

\bibitem{Arcaini2015}
Paolo Arcaini, Angelo Gargantini, and Elvinia Riccobene.
\newblock Using mutation to assess fault detection capability of model review.
\newblock {\em Softw. Test., Verif. Reliab.}, 25(5-7):629--652, 2015.

\bibitem{Arcaini2017}
Paolo Arcaini, Angelo Gargantini, and Elvinia Riccobene.
\newblock Nuseen: {A} tool framework for the nusmv model checker.
\newblock In {\em 2017 {IEEE} International Conference on Software Testing,
  Verification and Validation, {ICST} 2017, Tokyo, Japan, March 13-17, 2017},
  pages 476--483. {IEEE} Computer Society, 2017.

\bibitem{DBLP:conf/icst/BardinDDKPTM15}
S{\'{e}}bastien Bardin, Micka{\"{e}}l Delahaye, Robin David, Nikolai Kosmatov,
  Mike Papadakis, Yves~Le Traon, and Jean{-}Yves Marion.
\newblock Sound and quasi-complete detection of infeasible test requirements.
\newblock In {\em 8th {IEEE} International Conference on Software Testing,
  Verification and Validation, {ICST} 2015, Graz, Austria, April 13-17, 2015},
  pages 1--10, 2015.

\bibitem{DBLP:conf/icst/BardinKC14}
S{\'{e}}bastien Bardin, Nikolai Kosmatov, and Fran{\c{c}}ois Cheynier.
\newblock Efficient leveraging of symbolic execution to advanced coverage
  criteria.
\newblock In {\em Seventh {IEEE} International Conference on Software Testing,
  Verification and Validation, {ICST} 2014, March 31 2014-April 4, 2014,
  Cleveland, Ohio, {USA}}, pages 173--182, 2014.

\bibitem{Biere2011}
Armin Biere, Keijo Heljanko, and Siert Wieringa.
\newblock {AIGER} 1.9 and beyond, 2011.
\newblock Available at {\url{fmv.jku.at/hwmcc11/beyond1.pdf}}.

\bibitem{Boroday2007}
Sergiy Boroday, Alexandre Petrenko, and Roland Groz.
\newblock Can a model checker generate tests for non-deterministic systems?
\newblock {\em Electronic Notes in Theoretical Computer Science}, 190(2):3--19,
  2007.

\bibitem{Budd1979}
Timothy~A Budd, Richard~J Lipton, Richard~A DeMillo, and Frederick~G Sayward.
\newblock Mutation analysis.
\newblock Technical report, DTIC Document, 1979.

\bibitem{Cheng1993}
Szu-Tsung Cheng, Gary York, and Robert~K Brayton.
\newblock Vl2mv: A compiler from verilog to blif-mv.
\newblock {\em HSIS Distribution}, 1993.

\bibitem{Clarkson2014}
Michael~R. Clarkson, Bernd Finkbeiner, Masoud Koleini, Kristopher~K. Micinski,
  Markus~N. Rabe, and C{\'e}sar S{\'a}nchez.
\newblock {\em Temporal Logics for Hyperproperties}, pages 265--284.
\newblock Springer Berlin Heidelberg, Berlin, Heidelberg, 2014.

\bibitem{DBLP:journals/jcs/ClarksonS10}
Michael~R. Clarkson and Fred~B. Schneider.
\newblock Hyperproperties.
\newblock {\em Journal of Computer Security}, 18(6):1157--1210, 2010.

\bibitem{Fellner2017}
Andreas Fellner, Willibald Krenn, Rupert Schlick, Thorsten Tarrach, and Georg
  Weissenbacher.
\newblock Model-based, mutation-driven test case generation via
  heuristic-guided branching search.
\newblock In Jean{-}Pierre Talpin, Patricia Derler, and Klaus Schneider,
  editors, {\em Formal Methods and Models for System Design (MEMOCODE)}, pages
  56--66. ACM, 2017.

\bibitem{FinkbeinerHH18}
Bernd Finkbeiner, Christopher Hahn, and Tobias Hans.
\newblock Mghyper: Checking satisfiability of {HyperLTL} formulas beyond the
  {\(\exists^*\)} {\(\forall^*\)} fragment.
\newblock In Shuvendu~K. Lahiri and Chao Wang, editors, {\em Automated
  Technology for Verification and Analysis (ATVA)}, volume 11138 of {\em
  Lecture Notes in Computer Science}, pages 521--527. Springer, 2018.

\bibitem{Finkbeiner2015}
Bernd Finkbeiner, Markus~N. Rabe, and C{\'e}sar S{\'a}nchez.
\newblock Algorithms for model checking {HyperLTL} and {HyperCTL}$^*$.
\newblock In Daniel Kroening and Corina~S. P{\u{a}}s{\u{a}}reanu, editors, {\em
  Computer Aided Verification (CAV)}, Lecture Notes in Computer Science, pages
  30--48. Springer, 2015.

\bibitem{Fraser2009}
Gordon Fraser, Franz Wotawa, and Paul~E Ammann.
\newblock Testing with model checkers: a survey.
\newblock {\em Software Testing, Verification and Reliability}, 19(3):215--261,
  2009.

\bibitem{Gargantini1999}
Angelo Gargantini and Constance Heitmeyer.
\newblock Using model checking to generate tests from requirements
  specifications.
\newblock In {\em ACM SIGSOFT Software Engineering Notes}, volume~24, pages
  146--162. Springer-Verlag, 1999.

\bibitem{Hong2002}
Hyoung~Seok Hong, Insup Lee, Oleg Sokolsky, and Hasan Ural.
\newblock A temporal logic based theory of test coverage and generation.
\newblock In {\em International Conference on Tools and Algorithms for the
  Construction and Analysis of Systems}, pages 327--341. Springer, 2002.

\bibitem{Howden1982}
William~E. Howden.
\newblock Weak mutation testing and completeness of test sets.
\newblock {\em {IEEE} Trans. Software Eng.}, 8(4):371--379, 1982.

\bibitem{Lal2009}
Akash Lal and Thomas Reps.
\newblock Reducing concurrent analysis under a context bound to sequential
  analysis.
\newblock {\em Formal Methods in System Design}, 35(1):73--97, 2009.

\bibitem{DBLP:conf/icst/MarcozziDBKP17}
Micha{\"{e}}l Marcozzi, Micka{\"{e}}l Delahaye, S{\'{e}}bastien Bardin, Nikolai
  Kosmatov, and Virgile Prevosto.
\newblock Generic and effective specification of structural test objectives.
\newblock In {\em 2017 {IEEE} International Conference on Software Testing,
  Verification and Validation, {ICST} 2017, Tokyo, Japan, March 13-17, 2017},
  pages 436--441, 2017.

\bibitem{mcmillan92smv}
{McMillan, Kenneth L.}
\newblock The {SMV} system.
\newblock Technical Report CMU-CS-92-131, Carnegie Mellon University, 1992.

\bibitem{Nelson1989}
Greg Nelson.
\newblock A generalization of dijkstra's calculus.
\newblock {\em ACM Transactions on Programming Languages and Systems (TOPLAS)},
  11(4):517--561, October 1989.

\bibitem{Offutt1992}
A.~Jefferson Offutt.
\newblock Investigations of the software testing coupling effect.
\newblock {\em {ACM} Trans. Softw. Eng. Methodol.}, 1(1):5--20, 1992.

\bibitem{Rayadurgam2001}
Sanjai Rayadurgam and Mats Per~Erik Heimdahl.
\newblock Coverage based test-case generation using model checkers.
\newblock In {\em Engineering of Computer Based Systems (ECBS)}, pages 83--91.
  IEEE, 2001.

\bibitem{Tretmans1996}
Jan Tretmans.
\newblock Test generation with inputs, outputs and repetitive quiescence.
\newblock {\em Software - Concepts and Tools}, 17(3):103--120, 1996.

\bibitem{Visser2004}
Willem Visser, Corina~S Pǎsǎreanu, and Sarfraz Khurshid.
\newblock Test input generation with java pathfinder.
\newblock {\em ACM SIGSOFT Software Engineering Notes}, 29(4):97--107, 2004.

\bibitem{wang2018mualloy}
Kaiyuan Wang, Allison Sullivan, and Sarfraz Khurshid.
\newblock Mualloy: a mutation testing framework for alloy.
\newblock In {\em International Conference on Software Engineering: Companion
  (ICSE-Companion)}, pages 29--32. IEEE, 2018.

\end{thebibliography}

\clearpage

\appendix

\newpage
\section{Appendix}
\label{appendix:proofs}

\subsection{Verilog mutation operators}

\begin{table}[htb]
\centering
\caption{List of supported Verilog mutation operators ($^*$ marks bit-wise operations)}
\label{tab:mutop}

\begin{tabular}{p{2cm} p{7cm}}

\textbf{Type} & \textbf{Mutation} \\
\toprule
Arithmetic & Exchange binary $+$ and $-$  \\
           & Exchange unary  $+$ and $-$  \\
Relations  & Exchange $==$ and $!=$ \\
           & Exchange $<$, $\leq$, $>$, $\geq$  \\
Boolean    & Exchange $!$ and $\sim^*$  \\
           & Drop $!$ and $\sim^*$       \\
           & Exchange $\&\&$,$||$,$\&^*$,$|^*$, $XOR$ and $XNOR$ \\
Assignments & Exchange $=$ and $<=$  \\
            & (Blocking \& Non-Blocking Assignment)  \\
Constants & Replace Integer Constant $c$ by $0,1,c+1$, and $c-1$  \\
          & Replace Bit-Vector Constant by $\vec{0}$, and $\vec{1}$  \\
\bottomrule

\end{tabular}

\end{table}

\subsection{Mixed determinism cases}

\begin{lemma}
\label{lemma:satisfyparts}
Let $\Pi$ be a trace assignment with $p \defn \Pi(\pi)$, $q \defn \Pi(\pi')$, and $\STS^{c(m)}$ a conditional mutant.
\begin{enumerate}

\item $\Pi \models_{\STS^{c(m)}} \Box (\neg mut_{\pi})$ then $p\restrict{\InVars \cup \OutVars \cup \StateVars} \in \traces(\STS)$
\item $\Pi \models_{\STS^{c(m)}} \Box (mut_{\pi})$ then $p\restrict{\InVars \cup \OutVars \cup \StateVars} \in \traces(\STS^m)$
\item $\Pi \models_{\STS^{c(m)}} \Box (\bigwedge_{i \in \AP_{\InVars}} i_{\pi} \leftrightarrow i_{\pi'}))$ then $p\restrict{\InVars} = q\restrict{\InVars}$
\item $\Pi \models_{\STS^{c(m)}} \lozenge (\bigvee_{o \in \AP_{\OutVars}} \neg ( o_{\pi} \leftrightarrow o_{\pi'} ))$ then $p\restrict{\OutVars} \neq q\restrict{\OutVars}$

\end{enumerate}

\end{lemma}

\begin{proof}

The first two statements follow directly from the definition of conditional mutants.
The latter two statements follow directly from the fact that $\AP_{\InVars},\AP_{\OutVars}$ uniquely characterize inputs and outputs.

\end{proof}

\begin{restatable}{proposition}{dndpotential}
\label{thm:det_ndet_potential}
	Let the model $\STS$ with inputs $\InVars$  and outputs $\OutVars$  be deterministic and the mutant $\STS^m$ be non-deterministic.
	$$\STS^{c(m)} \models \phi_1(\InVars,\OutVars)\text{ iff } \STS^m\text{ is potentially killable.}$$
	
	\noindent Let $p$ be a $\pi$-witness for $\STS^{c(m)} \models \phi_1(\InVars,\OutVars)$, 
	then there is $n \in \mathbb{N}$, such that the test $t := p[0,n]\restrict{\InVars \cup \OutVars}$ potentially kills $\STS^m$.
\end{restatable}

\begin{proof}

Assume $\STS^m$ is potentially killable.
Let $q \in \traces(\STS^m)$, such that $q\restrict{I\cup O} \notin \traces(\STS)\restrict{I\cup O}$.
Since $\STS^m$ is equally input-enabled, 
there exists a trace $p \in \traces(\STS)$, such that $p\restrict{\InVars} = q\restrict{\InVars}$.
Clearly, $p\restrict{\OutVars} \neq q\restrict{\OutVars}$.
Therefore, $p$ and $q$ are satisfying assignments for $\Phi_1(\InVars,\OutVars)$ and $\pi$, $\pi'$ respectively.

Assume $\STS^{c(m)} \models \Phi_1(\InVars,\OutVars)$.
Let $p$ be a $\pi$-witness of $\Phi_1(\InVars,\OutVars)$ and 
let $q \in \traces(\STS^{c(m)})$ be a $\pi'$-witness of $\Phi_1(\InVars,\OutVars)$ without the first existential quantifier.
From Lemma~\ref{lemma:satisfyparts}, we immediately get $p\restrict{\InVars} = q\restrict{\InVars}$, and $p\restrict{\OutVars} \neq q\restrict{\OutVars}$.
This shows $\traces(\STS^m)\restrict{I\cup O} \nsubseteq \traces(\STS)\restrict{I\cup O}$.

Since $p\restrict{\OutVars} \neq q\restrict{\OutVars}$, there exists a smallest $n \in \mathbb{N}$ such that $p[0,n-1]\restrict{\OutVars} = q[0,n-1]\restrict{\OutVars}$ and $p[n]\restrict{\OutVars} \neq q[n]\restrict{\OutVars}$.
Clearly, $p[0,n]\restrict{I \cup O}$ potentially kills $\STS^m$.

\end{proof}

The following hyperproperty expresses definite killing for deterministic models and non-deterministic mutants:

\begin{equation}
	\label{eq:det_definite_kill}
	 \phi_4(\InVars,\OutVars) := \exists \pi \forall \pi' \Box (\neg mut_{\pi} \wedge mut_{\pi'}
	\bigwedge_{i \in \AP_{\InVars}} i_{\pi} \leftrightarrow i_{\pi'}) \rightarrow \lozenge (\bigvee_{o \in \AP_{\OutVars}} \neg ( o_{\pi} \leftrightarrow o_{\pi'} ))
\end{equation}

\begin{restatable}{proposition}{dnddefinite}
\label{thm:det_ndet_definite}
	Let the model $\STS$ with inputs $\InVars$  and outputs $\OutVars$  be deterministic and the mutant $\STS^m$ be non-deterministic.
	$$\STS^{c(m)} \models \phi_4(\InVars,\OutVars)\text{ iff } \STS^m \text{ is definitely killable.}$$
	
		\noindent If $t$ is a $\pi$-witness for $\STS^{c(m)} \models \phi_4(\InVars,\OutVars)$, then
 $t[0,n]\restrict{\InVars \cup \OutVars}$ definitely kills $\STS^m$ (for some $n\in\mathbb{N}$).
\end{restatable}

\begin{proof}

Assume that $\STS^m$ is definitely killable.
Since $\STS$ is deterministic, for every input sequence, there is at most one trace with in $\traces(\STS)$ with this input sequence.
Therefore, there is an input sequence $\vec{I}$ and a unique trace $p \in \traces(\STS)$ with $p_I = \vec{I}$ and $p\restrict{I\cup O} \notin \traces(\STS^m)\restrict{I \cup O}$.
Any trace assignment that maps $\pi$ to $p$ satisfies $\phi_4(\InVars,\OutVars)$, 
since either the antecedent is violated by a trace $q \in \traces(\STS^m)$ assigned to $\pi'$ with different inputs,
or the consequent is violated by a trace $q \in \traces(\STS^m)$ assigned to $\pi'$ with inputs $\vec{I}$ and 
outputs that can only be different to $p\restrict{\OutVars}$.

Conversely, assume $\STS^{c(m)} \models \phi_4(\InVars,\OutVars)$.
Let $q$ be a $\pi$-witness and $p$ be a $\pi'$-witness for which the antecedent is satisfied,
which is in fact satisfiable, since $\STS^m$ is equally input-enabled.
Then, from Lemma~\ref{lemma:satisfyparts}, we get
$p\restrict{\InVars \cup \OutVars \cup \StateVars} \in \traces(\STS)$,
$q\restrict{\InVars \cup \OutVars \cup \StateVars} \in \traces(\STS^m)$,
and $q\restrict{\InVars} = p\restrict{\InVars}$.
Since $\Pi$ satisfies the whole formula, 
it must be the case that $\Pi$ also satisfies the consequent, i.e. $q\restrict{\OutVars} \neq p\restrict{\OutVars}$ (Lemma~\ref{lemma:satisfyparts}).
Therefore, we can conclude $p\restrict{I\cup O} \notin \traces(\STS^m)\restrict{I \cup O}$, which, as argued above, 
is equivalent to definite killing in the deterministic model case.

The existence of the definitely killing test can be shown analogously as in the proof of Proposition~\ref{thm:det_ndet_potential}.

\end{proof}

\begin{restatable}{proposition}{nddkill}
\label{thm:ndet_det}
	Let the model $\STS$ with inputs $\InVars$  and outputs $\OutVars$  be non-deterministic and the mutant $\STS^m$ be deterministic
	$$\STS^{c(m)} \models \phi_2(\InVars,\OutVars) \text{ iff } \STS^m \text{ is killable}.$$
	
	\noindent Let $q$ be a $\pi$-witness for $\STS^{c(m)} \models \phi_2(\InVars,\OutVars)$, 
	then there is $n \in \mathbb{N}$, such that for any trace $p \in \traces(\STS)$ with $p\restrict{\InVars} = q\restrict{\InVars}$ the test $t := p[0,n]\restrict{\InVars \cup \OutVars}$ kills $\STS^m$.
\end{restatable}

\begin{proof}

Potential killing directly follows from the more restricted case in Proposition~\ref{thm:ndet_ndet_potential}.
Since $\STS^m$ is deterministic, by Proposition~\ref{prop:killability} it is also definitely killable.

\end{proof}

\subsection{Proofs of propositions}

\killability*

\begin{proof}

Let $\STS^m$ be potentially killable, which implies 
that there is a trace $q \in \traces(\STS^m)$, such that there is no trace $p \in \traces(\STS)$
with $q\restrict{I \cup O} = p\restrict{I \cup O}$, which implies $\traces(\STS^m)\restrict{I\cup O} \nsubseteq \traces(\STS)\restrict{I\cup O}$.

Let $\STS$ be deterministic and $\STS^m$ be potentially killable.
From the definition of determinism, it follows that for traces $q,q' \in \traces(\STS^m)$ with
$q_I = q'_I$ it is the case that $q = q'$.
In other words, for every sequence of inputs $\vec{I}$, 
we have $|\{q \in \traces(\STS^m) \mid q\restrict{I} = \vec{I}\}\restrict{\OutVars}| \leq 1$.
From potential killability (i.e. $\traces(\STS^m)\restrict{I\cup O} \nsubseteq \traces(\STS)\restrict{I\cup O}$) it follows that there exists $q \in \traces(\STS^m)$, such that
$q\restrict{\OutVars} \notin \{p \in \traces(\STS^m) \mid p\restrict{I} = q\restrict{I}\}\restrict{\OutVars}$.
Since the set of traces in the mutant sharing inputs with $q$ is a singleton, we have shown $\{q' \in \traces(\STS^m) \mid q'\restrict{I} = q\restrict{I}\}\restrict{\OutVars} \cap \{p \in \traces(\STS) \mid p\restrict{I} = q\restrict{I}\}\restrict{\OutVars} = \emptyset$, i.e. definite killability.
%
%
%
%

\end{proof}

\dkill*

\begin{proof}

Potential killing directly follows from the more restricted case in Proposition~\ref{thm:det_ndet_potential}.
Since $\STS^m$ is deterministic, by Proposition~\ref{prop:killability} it is also definitely killable.

\end{proof}

\ndpotential*

\begin{proof}

Assume that $\STS^m$ is potentially killable, which implies 
that there is a trace $q \in \traces(\STS^m)$, such that there is no trace $p \in \traces(\STS)$
with $q\restrict{I \cup O} = p\restrict{I \cup O}$.
Any trace assignment that maps $\pi$ to $p$ satisfies $\phi_2(\InVars,\OutVars)$, 
since either the antecedent is violated by a trace $q \in \traces(\STS^m)$ assigned to $\pi'$ with different inputs,
or the consequent is violated by a trace $q \in \traces(\STS^m)$ assigned to $\pi'$ with inputs $\vec{I}$ and 
outputs that can only be different to $p\restrict{\OutVars}$.

Conversely, assume $\STS^{c(m)} \models \phi_2(\InVars,\OutVars)$.
Let $p$ be a $\pi$-witness and $q$ be a $\pi'$-witness for which the antecedent is satisfied,
which is in fact satisfiable, since $\STS^m$ is equally input-enabled.
Then, from Lemma~\ref{lemma:satisfyparts}, we get
$p\restrict{\InVars \cup \OutVars \cup \StateVars} \in \traces(\STS)$,
$q\restrict{\InVars \cup \OutVars \cup \StateVars} \in \traces(\STS^m)$,
and $q\restrict{\InVars} = p\restrict{\InVars}$.
Since $\Pi[\pi \mapsto p, \pi' \mapsto q]$ satisfies the antecedent and the whole formula, 
the trace assignment also satisfies the consequent, i.e. $q\restrict{\OutVars} \neq p\restrict{\OutVars}$ (Lemma~\ref{lemma:satisfyparts}).
Since $q$ was chosen arbitrary (besides satisfying the antecedent), 
we can conclude $p\restrict{I\cup O} \notin \traces(\STS^m)\restrict{I \cup O}$, i.e. $\STS^m$ is potentially killable.

The existence of the potentially killing test can be shown analogously as in the proof of Proposition~\ref{thm:det_ndet_potential}.

\end{proof}

\nddefinite*

\begin{proof}

Let $\STS^m$ be definitely killable, which implies 
that there is a sequence of inputs $\vec{I} \in \traces(\STS)\restrict{\InVars}$, such that for $P_{\vec{I}} \defn \{p \in \traces(\STS) \mid p\restrict{\InVars} = \vec{I}\}$ and $Q_{\vec{I}} \defn \{q \in \traces(\STS) \mid q\restrict{\InVars} = \vec{I}\}$ it is the case that $P_{\vec{I}}\restrict{\OutVars} \cap Q_{\vec{I}}\restrict{\OutVars} = \emptyset$.
Since $\vec{I}$ is the input sequence of a trace of $\STS$, we have that $P_{\vec{I}} \neq \emptyset$.
Since $\STS^m$ is equally 
input-enabled, we have $Q_{\vec{I}} \neq \emptyset$.
We show that any $p \in P_{\vec{I}}$ is a $\pi$-witness to $\STS^{c(m)} \models \phi_3(\InVars,\OutVars)$.
Let $q' \in \traces(\STS^m)$ and $p'' \in \traces(\STS)$ be arbitrary traces and
consider a trace assignment that maps $\pi$ to $p$, $\pi'$ to $q'$ and $\pi''$ to $p''$
and assume that it satisfies the antecedent (which is satisfiable, due to $P_{\vec{I}} \neq \emptyset$ and $Q_{\vec{I}} \neq \emptyset$).
That is, $q' \in Q_{\vec{I}}$ and $p'' \in P_{\vec{I}}$.
Since $P_{\vec{I}}\restrict{\OutVars} \cap Q_{\vec{I}}\restrict{\OutVars} = \emptyset$,
we have shown $q'\restrict{\OutVars} \neq p''\restrict{\OutVars}$.
Since $q'$ and $p''$ were chosen arbitrarily, we have shown the trace assignment satisfies the formula,
i.e. $\STS^{c(m)} \phi_3(\InVars,\OutVars)$.

Conversely, assume $\STS^{c(m)} \models \phi_3(\InVars,\OutVars)$.
Let $p$ be a $\pi$-witness, $q'$ be a $\pi'$-witness, and $p''$ be a $\pi'$-witness for which the antecedent is satisfied,
which is in fact satisfiable, since $\STS^m$ is equally input-enabled.
Then, from Lemma~\ref{lemma:satisfyparts}, we get
$p\restrict{\InVars \cup \OutVars \cup \StateVars},p''\restrict{\InVars \cup \OutVars \cup \StateVars} \in \traces(\STS)$,
$q'\restrict{\InVars \cup \OutVars \cup \StateVars} \in \traces(\STS^m)$,
and $p\restrict{\InVars} = q'\restrict{\InVars} = p''\restrict{\InVars}$.

Since the $\Pi[\pi \mapsto p, \pi' \mapsto q',\pi'' \mapsto p'']$ satisfies the whole formula and the antecedent,
the trace assignment must also satisfy the consequent.
That is, it must be the case that $q'\restrict{\OutVars} \neq p''\restrict{\OutVars}$ (Lemma~\ref{lemma:satisfyparts}).
Since $q'$ and $p''$ were chosen arbitrarily (besides satisfying the antecedent),
we have shown $\{q \in \traces(\STS^m) \mid q\restrict{\InVars} = p\restrict{\InVars}\}\restrict{\OutVars} \cap \{p'' \in \traces(\STS) \mid p''\restrict{\InVars} = p\restrict{\InVars}\}\restrict{\OutVars} = \emptyset$, i.e. $\vec{I} \defn p\restrict{\InVars}$ is the input sequence showing that $\STS^m$ is definitely killable.

The existence of the definitely killing test can be shown analogously as in the proof of Proposition~\ref{thm:det_ndet_potential}.

\end{proof}

\transformationproperties*
%

\begin{proof}

We show $D(\STS^{c(m)})$ is deterministic (up to $\mutvar$).
$D(\STS^{c(m)})$ has a unique (up to $\mutvar$) initial state $X^{\tau}$ and initial output $O_{\varepsilon}$, 
since we fix $D(\alpha^{c(m)})$ to be satisfiable by exactly this state and output.

Assume that there are state $X$, successor $X'$ (both with respect to $\StateVars \cup \{mut\} \cup \{x^{\tau}\}$), inputs $I$ (with respect to  $\InVars \cup \{nd\}$), and outputs $O_1,O_2$,
$X \trans{I,O_1} X'_1$, $X \trans{I,O_2} X'_2$ and either $O_1 \neq O_2$ or $X'_1\restrict{\StateVars} \neq X'_2 \restrict{\StateVars}$.
That is, there is real non-determinism remaining in $D(\STS^{c(m)})$.
First we note that $X'_1(x^{\tau}) = X'_2(x^{\tau}) = \bot$ due to the constraint we add to $D(\delta^{c(m)})$ for $x^{\tau}$,
so we can rule out difference in valuation of this variable.
Since output $O$ and successor state $X'$ are uniquely fixed with input $nd(O,I)$,
it needs to be the case that $X \trans{I\restrict{\InVars \setminus \{nd\}},O_1} X'_1\restrict{\StateVars_+}$, $X \trans{I\restrict{\InVars \setminus \{nd\}},O_2} X'_2\restrict{\StateVars_+}$.
This implies  that the transitions are also present in $\STS^{c(m)}$.
That is, $(X,I)$ causes non-determinism in $\STS^{c(m)}$.
Therefore, in the transformation of $\STS^{c(m)}$, additional constraints are added, 
tying outputs and successors of $(X,I)$ to unique values of $nd$,
which is a contradiction to the statement that the transitions exist irrespective to the value of $nd$.

We show $p \in \traces(\STS^{c(m)})\restrict{\StateVars_+ \cup \InVars \cup \OutVars}[0,n]$ then $p \in \traces(D(\STS^{c(m)}))[1,n+1]\restrict{\StateVars_+ \cup \InVars \cup \OutVars}$ by induction on $n$.

First note that sets $J^{\cap}$, $J^{orig}$, and $J^{mut}$ are pairwise disjoint and contain every initial state/output pair of $\STS^{c(m)}$,
since use the very definition of initial state/output pairs, possibly splitting them according to values of $mut$.
Likewise, for each $(X,I)$ sets $T^{\cap}_{(X,I)} \cup T^{mut}_{(X,I)} \cup T^{orig}_{(X,I)}$
are pairwise disjoint and contain every transition for $(X,I)$ as a tuple $(X,I,O,X')$.

In the base case, $n = 0$, let $p = (I_0,X_0,O_0)$, where $X_0$ and 
$O_0$ are initial state and output of $\STS^{c(m)}$.
We need to show that there is a trace $q \in \traces(D(\STS^{c(m)}))$, 
such that $q[1]\restrict{\StateVars_m \cup \InVars \cup \OutVars} = (I_0,X_0,O_0)$.
As noted above $(O_0,X_0) \in J^{\cap} \cup J^{orig} \cup J^{mut}$.
Therefore, we add a constraint corresponding to a transition 
$X^{\tau} \trans{nd(O_0,X_0),O_0} X'_0$ to the system.
Furthermore, for $(X_0,I_0)$, we add a constraint corresponding to a transition
$X_0 \trans{I_0,O} X'$ for some output $O$ and successor $X'$.
Therefore, the trace $q$ exists.

In the inductive step, assume that the statement holds for $n-1$ and consider the case for $n$.
Let $p[n-1] = (I_{n-1},X_{n-1},O_{n-1})$ and $p[n] = (I_{n},X_{n},O_{n})$.
We need to show that for some trace $q \in \traces(D(\STS^{c(m)}))$ 
with $q[n]\restrict{\StateVars_m \cup \InVars \cup \OutVars} = (I_{n-1},X_{n-1},O_{n-1})$, 
which exists due to the induction hypothesis, 
it is the case that $q\restrict{\StateVars_m \cup \InVars \cup \OutVars}[n+1] = (I_{n},X_{n},O_{n})$.
Again, this is a consequence of $(X_{n-1},I_{n-1},O_n,X_{n}') \in T^{\cap}_{(X_{n-1},I_{n-1})} \cup T^{mut}_{(X_{n-1},I_{n-1})} \cup T^{orig}_{(X_{n-1},I_{n-1})}$
and the transitions we add to the system, exhaustively constrained by values of $nd$ for output/successor pairs.

Note also that a consequence of the above statements, and the fact that we introduce transitions for different values of $nd$ exhaustively,
is that $D(\STS^{c(m)})$ preserves equal input-enabledness.

$\not\models \phi_1(\InVars,\OutVars)$ then $\STS^m$ is equivalent is a direct consequence of the statements about traces,
since $\not\models \phi_1(\InVars,\OutVars)$ shows no trace in $\traces(D(\STS^{c(m)}))$ is a witness to killing the mutant.
Since traces of $\STS^{c(m)}$ are included in (the projection of) this set, 
there can not be a trace in $\traces(\STS^{c(m)})$ that is a witness to killing the mutant.
\end{proof}

\end{document}